\PassOptionsToPackage{
pdfencoding=auto,
pdfnewwindow=true,
pdfusetitle=true,
bookmarks=true,
bookmarksnumbered=true,
bookmarksopen=true,
pdfpagemode=UseThumbs,
bookmarksopenlevel=1,
pdfpagelabels=false,
breaklinks=true,
colorlinks=true
}{hyperref}
\pdfoutput=1
\PassOptionsToPackage{usenames,dvipsnames}{xcolor}
\PassOptionsToPackage{square,sort&compress,comma,numbers,elide}{natbib}
\PassOptionsToPackage{justification=justified,singlelinecheck=false,textfont={small},labelfont={small,bf},margin=0.5in}{caption}
\documentclass[onecolumn,11pt,letterpaper, final, accepted = 2018-08-17]{quantumarticle}
\usepackage{natbib}

\usepackage{amsmath, amsfonts, amssymb, enumerate, amsthm, graphicx, braket, relsize, bbm, mathrsfs, tikz}
\usetikzlibrary{calc,shapes,decorations,decorations.pathreplacing,patterns,arrows,shapes}
\usepackage[usenames,dvipsnames]{xcolor} 

\usepackage{mathtools, paralist}
\usepackage[justification=justified,singlelinecheck=false,textfont={small},labelfont={small,bf}]{subcaption}
\usepackage{ragged2e}
\DeclareCaptionJustification{justified}{\justifying}
\usepackage{soul} 
\usepackage{units}
\usepackage[utf8]{inputenx}
\usepackage[OT1]{fontenc}
\usepackage[english]{babel}
\usepackage{csquotes}

\definecolor{purple}{RGB}{128,0,128}
\definecolor{ultramarine}{RGB}{63, 0, 255}
\definecolor{medblue}{RGB}{0, 0, 100}
\definecolor{googleblue}{RGB}{34, 0, 204}
\definecolor{panblue}{RGB}{0,24,150}
\definecolor{carmine}{RGB}{150, 0, 24}
\definecolor{gray}{RGB}{150, 150, 150}
\definecolor{darkgreen}{RGB}{0, 80, 0}

\usepackage{hyperref}
\hypersetup{linkcolor=carmine,citecolor=darkgreen,urlcolor=googleblue,anchorcolor=OliveGreen,}

\newtheorem{thm}{Theorem}
\newtheorem{thmrepeat}{Theorem}

\newtheorem{lemma}{Lemma}[thm]
\newtheorem{cor}{Corollary}[thm]
\newtheorem*{conj}{Conjecture}

\newtheoremstyle{defblock}{0.7\topsep}{0pt}{}{}{\bfseries}{: }{0pt plus 1pt minus 1pt}{\thmname{\bfseries{#1}}\thmnumber{\bfseries{#2}}\thmnote{#3}}
\theoremstyle{defblock}
\newtheorem*{defn}{}
\theoremstyle{definition}
\newtheorem{ex}{Example}
\newtheorem*{inductionhypothesis}{Induction hypothesis (IH$_{\text{m}}$)}
\newtheorem*{basecase}{Base case}
\newtheorem*{inductivestep}{Inductive step}

\theoremstyle{remark}

\newtheoremstyle{claim}{0.7\topsep}{0pt}{}{}{\bfseries}{: }{0pt plus 1pt minus 1pt}{\thmname{#1} \thmnumber{{#2}}\thmnote{\normalfont#3}}
\theoremstyle{claim}
\newtheorem{claim}{Claim}[lemma]

\newtheoremstyle{block}{0.7\topsep}{0pt}{\em}{}{\bfseries}{: }{0pt}{\thmnote{\bfseries#3}\thmnumber{{#2}}}
\theoremstyle{block}
\newtheorem*{block}{}

\newcommand{\M}{\mathcal{M}}
\newcommand{\x}{\mathbf{x}}
\newcommand{\Q}{\mathcal{Q}}
\renewcommand{\C}{\mathcal{C}}
\newcommand{\CE}{\mathcal{CE}}
\renewcommand{\G}{\mathcal{G}}
\newcommand{\I}{\mathcal{I}}
\newcommand{\cH}{\mathcal{H}}
\newcommand{\Id}{\mathbbm{1}}
\newcommand{\cP}{\mathcal{P}}
\newcommand{\cB}{\mathcal{B}}
\newcommand{\cD}{\mathcal{D}}
\newcommand{\ev}[1]{\langle #1 \rangle}
\newcommand{\cS}{\mathscr{S}}
\newcommand{\scalemath}[2]{\scalebox{#1}{\mbox{\ensuremath{\displaystyle #2}}}}
\renewcommand{\tfrac}[2]{\vcenter{\hbox{$\genfrac{}{}{}{1}{#1}{#2}$}}}
\newcommand{\subscript}[1]{$_{\text{#1}}$}
\newcommand{\tr}[1]{\mathrm{tr} \left[ #1 \right]}

\newcommand{\hyphenationsetting}{%
	\emergencystretch=0pt 
	\tolerance=2000 
	\pretolerance=1000 
	\righthyphenmin=4 
	\lefthyphenmin=4 
	}
\newcommand{\setmuskip}[2]{#1=#2\relax}
\newcommand{\narrowmath}[1]{%
	\begingroup%
		\setmuskip{\thinmuskip}{3mu}%
		\setmuskip{\medmuskip}{1.5mu}%
		\setmuskip{\thickmuskip}{5mu}%
		\ensuremath{#1}%
	\endgroup 
	\hyphenationsetting} 

\newcommand{\bel}[1]{{\color{red!70!blue} #1}}

\let\OLDthebibliography\thebibliography
\renewcommand\thebibliography[1]{
	\OLDthebibliography{#1}
	\setlength{\parskip}{0pt}
	\setlength{\itemsep}{0pt plus 0.3ex}
	}

\newenvironment{sloppypar*}{\sloppy\ignorespaces}{\par}

\hyphenationsetting

\MakeOuterQuote{"} 

\DeclareFontFamily{OT1}{cmrx}{}
\DeclareFontShape{OT1}{cmrx}{m}{n}{<->cmr10}{}
\renewcommand*{\Longrightarrow}{\mathrel{\mbox{\fontfamily{cmrx}\fontencoding{OT1}\selectfont=}}\joinrel\Rightarrow}

\begin{document}

\title{Almost Quantum Correlations are Inconsistent with Specker's Principle}
\author{Tomáš~Gonda}
\affiliation{Perimeter Institute for Theoretical Physics, Waterloo, Ontario, Canada, N2L 2Y5}
\affiliation{Dept.~of Physics and Astronomy, University of Waterloo, Waterloo, Ontario, Canada, N2L 3G1}
\email{tgonda@perimeterinstitute.ca}
\author{Ravi~Kunjwal}
\affiliation{Perimeter Institute for Theoretical Physics, Waterloo, Ontario, Canada, N2L 2Y5}
\author{David~Schmid}
\affiliation{Perimeter Institute for Theoretical Physics, Waterloo, Ontario, Canada, N2L 2Y5}
\affiliation{Dept.~of Physics and Astronomy, University of Waterloo, Waterloo, Ontario, Canada, N2L 3G1}
\author{Elie~Wolfe}
\author{Ana~Belén~Sainz}
\affiliation{Perimeter Institute for Theoretical Physics, Waterloo, Ontario, Canada, N2L 2Y5}

\date{\today}

\begin{abstract}
	\noindent Ernst Specker considered a particular feature of quantum theory to be especially fundamental, namely that pairwise joint measurability of sharp measurements implies their global joint measurability [\href[pdfnewwindow]{https://vimeo.com/52923835}{vimeo.com/52923835 (2009)}]. To date, Specker's principle seemed incapable of singling out quantum theory from the space of all general probabilistic theories. In particular, its well-known consequence for experimental statistics, the principle of consistent exclusivity, does not rule {out the set of correlations known as \textit{almost quantum}, which is} strictly larger than the set of quantum correlations. Here we show that, contrary to the popular belief, Specker's principle \emph{cannot} be satisfied in any theory that yields almost quantum correlations.
\end{abstract}

\maketitle

\section{Introduction}
	\label{sec:Introduction}
	
	The advent of quantum theory was accompanied by many conceptual controversies over the failure of intuitions from classical physics, e.g.~the existence of wave-particle duality, the fundamental indeterminism apparent from the Born rule, and the nonseparability epitomized by entanglement, as pointed out by Einstein, Podolsky and Rosen~\cite{solvay, EPR}. More recently, we have witnessed the emergence of quantum information theory and a surge of interest in quantum foundations. As a consequence, there has been remarkable progress in proving theorems concerning ways in which Nature fails to be classical, given a well-defined notion of classicality that mathematically formalizes some intuition from classical physics.
	
	This attitude can also be taken towards quantum theory. Are there ways in which Nature may fail to be quantum? That is, are there deviations from quantum theory in Nature, and if so, where should one look for them~\cite{GPTexpt}? Recently, there has been much effort to identify the physical properties that single out the quantum world from the space of hypothetical alternatives. By doing this, not only can we learn more about quantum theory itself, but we also gain insight about where one might or might not hope to find failures of quantum theory.
	
	In the past two decades, there have been two major lines of research tackling this problem. On the one hand, there is the program of characterizing only the \emph{statistical} aspect of quantum theory, i.e.~recovering quantum correlations. On the other hand, there is that of deriving the \emph{structure} of quantum theory from a set of simple axioms.
	
	Statistical aspects of quantum theory come into play when exploring phenomena such as Bell nonlocality~\cite{Bell64} and Kochen-Specker (KS) contextuality~\cite{KS}. These cannot be explained by a classical model of the world, although they arise naturally within quantum theory. Despite being physically distinct \cite{Spe05}, the notions of classicality challenged by these phenomena are mathematically similar, which makes it possible to study them in a unified manner \cite{AB,AFLS}. 
	
	Tackling the statistical aspects of quantum theory has provided us with considerable progress in articulating principles that are satisfied by quantum correlations. However, a key question remains. Are there any principles that \emph{uniquely} identify the set of quantum correlations in the space of all conceivable correlations? The study of this question led to the conception of a set known as almost quantum correlations~\cite{AQ}. Initially defined only within Bell scenarios~\cite{AQ}, almost quantum correlations have also been subsequently defined within general (KS-)contextuality scenarios~\cite{AFLS}. This set of correlations strictly contains the quantum set, yet has the following remarkable property: Almost quantum correlations are consistent with all the principles that to our knowledge have so far been proposed to characterize quantum correlations\footnote{Strictly speaking, almost quantum correlations have not been proven to satisfy the principle of Information Causality yet; only numerical evidence exists for that claim.}~\cite{AQ}. In particular, almost quantum correlations satisfy the principle of consistent exclusivity (CE) \cite{AFLS}, which we define at the end of section~\ref{se:probmod}. Therefore, it is desirable to identify a principle capable of discriminating quantum from almost quantum correlations. 
	
	The other approach deals with structural aspects of quantum theory. This has proven more successful than the previous approach {since, to date, there are already} derivations of quantum theory from a set of reasonably simple axioms~\cite{AxioLucien, AxioBook2, AxioMarkus, AxioLluis, Muller2016, Chiribella2016, Hardy2016, Dakic2016, AxioOeckl, hoehn1, hoehn2}. Hence, lately, there has been increased interest in how this successful approach could be used to address the aforementioned issues with almost quantum correlations. The hope is to clarify the relevant structural differences between {quantum theory and} physical theories capable of generating almost quantum correlations. For example, recent results show that any general probabilistic theory~\cite{Barrett} giving rise to almost quantum correlations would have to violate the no-restriction hypothesis~\cite{giulionoR, noR}. 
	
	\bigskip
	
	In this paper, we explore the structural aspects of a hypothetical almost quantum theory by investigating its statistical aspects---almost quantum correlations. To this end, we use the connections between two frameworks for describing contextuality scenarios: in terms of compatible measurements \cite{AB}, and in terms of operational equivalences among measurement events \cite{AFLS}. First, in section \ref{sec:2}, we review these two formalisms and their relation.	Section \ref{sec:3} then presents Specker's principle and some of its consequences both for the structure of a theory and the statistics arising from it. In particular, these include a family of sufficiency properties which we define therein. In section \ref{sec:4}, we show that the ramifications of one of these sufficiency properties for measurement statistics are violated by almost quantum correlations. Thus, Specker's principle cannot be satisfied in any theory yielding almost quantum correlations. Finally, section~\ref{sec:5} provides evidence that not all of the sufficiency properties are in conflict with almost quantum correlations, demonstrating the subtle nature of the contradiction between these and Specker's principle.

\section{Marginal Scenarios and the Event-Based Hypergraph Approach to Contextuality}
	\label{sec:2}
	
	Contextuality has been studied in a number of different formalisms~\cite{AB, CSW, AFLS}. The two that we are interested in here are traditional contextuality scenarios defined from compatibility structures of measurements, as formalised by Abramsky and Brandenburger in \cite{AB}, and the more general approach\footnote{For the particular contextuality scenarios considered in this work, there is no extra generality in the hypergraph formalism of \cite{AFLS} compared with the sheaf-theoretic formalism of \cite{AB}.} by Ac\'in, Fritz, Leverrier and Sainz~\cite{AFLS}. In this section, we lay out the key concepts of both formalisms, their connections, and important sets of {correlations (a.k.a.~probabilistic models).} 
	We begin by making the notion of compatible measurements explicit.
	
	\subsection{Compatibility of Measurements}
		\label{sec:2-Compatibility}
		
		\begin{defn}[Compatible measurements]{\cite{notesonjm,LSW}}\label{def:CompMeas}\\
			Let $\{A_1, \ldots, A_n\}$ be a set of $n$ measurements with the corresponding outcome sets being $\{O_1, \ldots, O_n\}$. We say that the elements of the set $\{A_1, \ldots, A_n\}$ are \emph{compatible} (or jointly measurable) if there exists a \mbox{measurement $A^{\prime}$} with outcome set $O_1 \times \ldots \times O_n$, such that all the measurements in the set $\{A_1, \ldots, A_n\}$ can be simulated by post-processing the outcomes of $A^{\prime}$. That is, the statistics of $\{A_1, \ldots, A_n\}$ for every preparation\footnote{Although we denote a preparation with the symbol $\rho$, we are not assuming here that the preparation is a \emph{quantum} state.} $\rho$ can be reconstructed perfectly by measuring the given preparation $\rho$ with $A^{\prime}$:
			\begin{align}
				\operatorname{Pr}\,(a_i \!\mid\! A_i,\rho) = \sum\limits_{\alpha\setminus a_i} \operatorname{Pr}\,(\alpha \!\mid\! A^{\prime},\rho)	\qquad \forall\, \rho, \; \forall\,i\in\{1, \ldots, n\},
			\end{align}
			where $\alpha \in O_1 \times \ldots \times O_n$, and $a_i \in O_i$ is fixed in the sum on the right hand side.
		\end{defn}
		
		\bigskip
		
		As an example, consider two compatible measurements $A_1$ and $A_2$ in quantum theory.
		\begin{align}
			A_1 = \left\{A_1^0, A_1^1 \right\} &:= \left\{\tfrac{1}{2} \left( \mathbbm{1} + \tfrac{1}{\sqrt{2}} \, \sigma_x \right), \tfrac{1}{2} \left( \mathbbm{1} - \tfrac{1}{\sqrt{2}} \, \sigma_x \right) \right\} , \\
			A_2 = \left\{A_2^0, A_2^1 \right\} &:= \left\{\tfrac{1}{2} \left( \mathbbm{1} + \tfrac{1}{\sqrt{2}} \, \sigma_z \right), \tfrac{1}{2} \left( \mathbbm{1} - \tfrac{1}{\sqrt{2}} \, \sigma_z \right) \right\} ,
		\end{align}
		where $\sigma_x$ and $\sigma_z$ are two of the Pauli matrices. The outcome set for each measurement is $O = \{0,1\}$ in this case. If by $\rho$ we denote the (arbitrary) state of a quantum system, then for any $a_1,a_2 \in O$,
		\begin{align}
			\mathrm{Pr} \left(a_1 \!\mid\! A_1,\rho \right) = \tr{\tfrac{1}{2} \left( \mathbbm{1} + \tfrac{(-1)^{a_1}}{\sqrt{2}} \, \sigma_x \right) \rho} , \\
			\mathrm{Pr} \left( a_2 \!\mid\! A_2,\rho \right) = \tr{\tfrac{1}{2} \left( \mathbbm{1} + \tfrac{(-1)^{a_2}}{\sqrt{2}} \, \sigma_z \right) \rho} .
		\end{align}
		 In order to show that $A_1$ and $A_2$ are compatible, we have to construct a measurement that can simulate them both. One such measurement is $A^{\prime} = \{N_{a_1, a_2} \}$, with $N_{a_1,a_2}$ defined for any $a_1,a_2 \in O$ as follows:
		\begin{align}
			N_{a_1,a_2} := \tfrac{1}{4} \left( \mathbbm{1} + \frac{(-1)^{a_1} \sigma_x + (-1)^{a_2} \sigma_z }{\sqrt{2}} \right)
		\end{align}
		This measurement has the following properties:
		\begin{align}
			A_1^{a_1} &= \sum_{a_2 \in O} N_{a_1,a_2} &  A_2^{a_2} &= \sum_{a_1 \in O} N_{a_1,a_2}
		\end{align}
		It follows that, for all $\rho$, the probabilities of outcomes of $A_1$ and of $A_2$ can be simulated by coarse-graining the probabilities of outcomes of $A'$:
		\begin{align}
			\mathrm{Pr} \left(a_1 \!\mid\! A_1,\rho \right) = \sum_{a_2} \tr{N_{a_1,a_2} \, \rho} = \sum_{a_2} \mathrm{Pr} \left(a_1,a_2 \!\mid\! A^{\prime},\rho \right) , \\
			\mathrm{Pr} \left(a_2 \!\mid\! A_2,\rho \right) = \sum_{a_1} \tr{N_{a_1,a_2} \, \rho} = \sum_{a_1} \mathrm{Pr} \left(a_1,a_2 \!\mid\! A^{\prime},\rho \right) ,
		\end{align}
		and thus $A_1$ and $A_2$ are compatible.
		
		\bigskip
		
		As a consequence of the definition of compatibility presented above, the marginal statistics of the set of compatible measurements has to be consistent with the existence of a joint probability distribution over $O_1 \times \ldots \times O_n$. Equivalently, whenever the statistics of a set of measurements fails to admit a joint probability distribution, there must necessarily be incompatibility within that set of measurements.
		
		However, it is important to point out that the converse need not hold. That is, if for each preparation there exists a joint probability distribution over $O_1 \times \ldots \times O_n$ that gives the statistics of each of the measurements $\{A_1, \ldots, A_n\}$, it does not follow that these are compatible measurements. The simplest counter-example is that of two projective measurements in quantum theory. Obviously, one can construct a joint distribution by taking the product of the respective distributions for each of the two measurements, regardless of whether they are compatible (i.e.~commuting) or not. These considerations will be important for understanding the relations that hold among the sufficiency properties in section~\ref{sec:3}.
	
	\subsection{Compatibility Approach to Contextuality}\label{se:comp}
		
		The traditional compatibility approach~\cite{AB}, also known as sheaf-theoretic approach, is concerned with a set of measurements and the compatibility relations among them. These relations are described by \emph{contexts}, each of which is a set of compatible measurements. A contextuality scenario is defined by a triplet $(X,O,\mathcal{M})$ and is called a \emph{marginal scenario}\footnote{Some works in the literature also refer to these scenarios as \emph{measurement scenarios}. Here we will stick to the notation used in \cite{AFLS}.} in this approach. $X$ denotes the set of measurements, each of whose set of outcome labels is without loss of generality assumed to be identical and equal to $O$. On the other hand, $\M$ is a subset of the power set of $X$ that corresponds to the \emph{maximal} contexts, i.e.~those contexts which are not contained within any other one. Hence, all the contexts in a given scenario are precisely the subsets of elements of $\M$. An example of the marginal scenario usually referred to as \emph{Specker's triangle} is presented in figure \ref{fig:32}. 
		
		{Given a maximal context $\x \in \M$, let $O^{\x}$ denote the set of all families of outcomes of the measurements in $\x$. That is, $O^{\x}$ naturally corresponds to the set of functions $\x \to O$. An assignment of probabilities to measurement outcomes, \mbox{$\{ P_{\x}(\mathbf{a}) \,:\, \x \in \M ,\, \mathbf{a} \in O^{\x} \}$}, is called an \emph{empirical model}.} For the purposes of studying the phenomenon of contextuality, we focus on empirical models that satisfy the no-disturbance principle, which states that for any two contexts $\x, \x' \in \mathcal{M}$,  $P_{\x|\x \cap \x^\prime} = P_{\x^\prime|\x \cap \x^\prime}$ holds. Here, $P_{\x|\x \cap \x^\prime}$ denotes the marginal distribution of $P_{\x}$ associated to the measurements in $\x \cap \x^\prime$. In words, no-disturbance imposes that the marginal statistics of a subset of measurements does not depend on the context in which they could be measured. No-disturbance also provides justification for considering only maximal contexts: If we are interested in the statistics of a set of compatible measurements that do not form a maximal context, we can use any maximal context that contains all of them and get the same result.
		
		An empirical model $P$ is said to be \emph{quantum} whenever its statistics may be recovered by performing, on some quantum state, projective measurements which satisfy the given compatibility relations.
		
		There is a particular family of marginal scenarios, mentioned also in \cite{ruishane}, that is of interest to us with regards to the discussion in this paper. We call them symmetric marginal scenarios and they are characterized by two parameters $(n,k)$.
		
		\begin{defn}[Symmetric marginal scenarios]{\hspace{1pt}}\label{def:SCS}\newline
			An $(n,k)$ \emph{symmetric marginal scenario} is a triple $(X,O,\mathcal{M})$ with $|X| = n$ and 
			\begin{equation}
				\mathcal{M} = \{\Omega \subseteq X \,:\, |\Omega| = k \}
			\end{equation}
		\end{defn}
	
	\subsection{Events-Based Hypergraph Approach to Contextuality}\label{se:ebh}
		
		The second formalism we discuss is the hypergraph-based formalism from reference~\cite{AFLS}. Here, a contextuality scenario is defined as an \emph{events-based hypergraph} $H = (V, E)$ with vertices $V$ and hyperedges $E$. The vertices correspond to the events in the scenario, each of which represents an outcome obtained from a device after it receives an input---the measurement choice. The hyperedges are sets of events representing all the possible outcomes given a particular measurement choice. The hypergraph approach assumes that every such measurement set is complete. That is, if the measurement corresponding to hyperedge $e$ is performed, exactly one of the outcomes in $e$ is obtained. Note that measurement sets may have non-trivial intersection. If an event appears in more than one hyperedge, it corresponds to outcomes of implementing different measurement choices that are operationally equivalent\footnote{Two measurement events are operationally equivalent if their probabilities coincide for every preparation of the system being probed \cite{Spe05}.}. 
		
		There is a close connection between the two approaches. In particular, given a marginal scenario $(X,O,\M)$, we can define an events-based hypergraph $H[X,O,\M]$ that corresponds to the same situation, in accordance with the Appendix D of~\cite{AFLS}. In a nutshell, the construction works as follows. 
		
		The vertices of $H[X,O,\M]$ are defined as $V(H) = \left\{ (\mathbf{a}|\mathbf{x}) : \mathbf{x} \in \mathcal{M}, \, \mathbf{a} \in O^\mathbf{x} \right\}$, where $\mathbf{x} \in \mathcal{M}$ is a maximal context and $\mathbf{a} \in O^\mathbf{x}$ is a family of outcomes of the measurements in $\mathbf{x}$ as before. The vertices of $H[X,O,\M]$ thus inherit the underlying structure of the initial scenario, since their labels contain the information about the maximal contexts and the possible outcomes. 
		
		On the other hand, the hyperedges arise via measurements protocols \citep[Def.~D.1.4]{AFLS} (see next paragraph for an example), which correspond to adaptive choices of measurements from $X$. They are adaptive, because a choice of an individual element of $X$ can depend on the outcomes of the previously chosen ones.

		\begin{figure}[!tb]
			\begin{center}
				\begin{subfigure}[b]{.4\textwidth}\centering
					\begin{tikzpicture}
						\node[draw=gray!60!black,fill, color=gray!60!black, shape=rectangle,scale=.5] at (-30:1) {} ;
						\node[draw=gray!60!black,fill, color=gray!60!black, shape=rectangle,scale=.5] at (90:1) {} ;
						\node[draw=gray!60!black,fill, color=gray!60!black, shape=rectangle,scale=.5] at (210:1) {} ;
						\path (-30:1) -- (90:1) node[pos=.5] (a) {};
						\path (90:1) -- (210:1) node[pos=.5] (b) {};
						\path (210:1) -- (-30:1) node[pos=.5] (c) {};
						\draw[thick,color=gray!60!black, rotate around={-60:(a)}] (a) ellipse (1.4cm and .3cm) ;
						\draw[thick,color=gray!60!black, rotate around={60:(b)}] (b) ellipse (1.4cm and .3cm) ;
						\draw[thick,color=gray!60!black] (c) ellipse (1.4cm and .3cm) ;
						\node at (-30:1.7) {$A_2$};
						\node at (90:1.7) {$A_1$};
						\node at (210:1.7) {$A_3$};
						\node at (0,-2) {};
					\end{tikzpicture}
					\caption{Specker's triangle as a marginal scenario.}\label{fig:32}
				\end{subfigure}
				\begin{subfigure}[b]{.4\textwidth}\centering
					\centering
					\begin{tikzpicture}[scale=1.5]
						\foreach \a in {0} \foreach \b in {0,1}
						{
						\node[draw,fill, color=gray!60!black, shape=circle,scale=.4] at (1+\b,1-\a) {};
						\node[above=7pt] at (1+\b,1-\a) {\tiny{$(\a\b|13)$}} ;
						\node[draw,fill, color=gray!60!black, shape=circle,scale=.4] at (-1+\b,1-\a) {};
						\node[above=7pt] at (-1+\b,1-\a) {\tiny{$(\a\b|12)$}} ;
						\node[draw,fill, color=gray!60!black, shape=circle,scale=.4] at (\b,-1-\a) {};
						\node[above=7pt] at (\b,-1-\a) {\tiny{$(\a\b|23)$}} ;
						}
						\foreach \a in {1} \foreach \b in {0,1}
						{
						\node[draw,fill, color=gray!60!black, shape=circle,scale=.4] at (1+\b,1-\a) {};
						\node[below=7pt] at (1+\b,1-\a) {\tiny{$(\a\b|13)$}} ;
						\node[draw,fill, color=gray!60!black, shape=circle,scale=.4] at (-1+\b,1-\a) {};
						\node[below=7pt] at (-1+\b,1-\a) {\tiny{$(\a\b|12)$}} ;
						\node[draw,fill, color=gray!60!black, shape=circle,scale=.4] at (\b,-1-\a) {};
						\node[below=7pt] at (\b,-1-\a) {\tiny{$(\a\b|23)$}} ;
						}
						\draw[thick,color=blue!70!black,rounded corners] (-1.2,-0.2) -- (-1.2,1.2)--(0.2,1.2)--(0.2,-0.2)--cycle;
						\draw[thick,color=blue!70!black,rounded corners] (0.8,-0.2) -- (0.8,1.2)--(2.2,1.2)--(2.2,-0.2)--cycle;
						\draw[thick,color=blue!70!black,rounded corners] (-0.2,-2.2) -- (-0.2,-0.8)--(1.2,-0.8)--(1.2,-2.2)--cycle;
						\draw[thick,color=Orange!70!black,rounded corners] (-1.25,1.15)--(0.25,1.15)-- (0.75,0.15) -- (2.25,0.15) -- (2.25,-0.15) -- (0.75,-0.15) -- (0.25, 0.85) -- (-1.25,0.85) --cycle;
						\draw[thick,color=Dandelion!70!black,rounded corners] (-1.25,0.15)--(0.25,0.15)-- (0.75,1.15) -- (2.25,1.15) -- (2.25,0.85) -- (0.75,0.85) -- (0.25, -0.15) -- (-1.25,-0.15) --cycle;
						\draw[thick,color=green!70!black,rounded corners] (-1.15, 1.15) -- (-0.85,1.15) -- (-0.85, 0) -- (-0.15,-1.85) -- (1.15,-1.85) -- (1.15,-2.15) -- (-0.15,-2.15) -- (-1.15,-0.15) -- cycle;
						\draw[thick,color=PineGreen!70!black,rounded corners] (-0.15,1.15) -- (0.15,1.15) -- (0.15,-0.85) -- (1.15,-0.85) -- (1.15,-1.15) -- (-0.15,-1.15) -- cycle;
						\draw[thick,color=Mulberry!70!black,rounded corners] (0.85,1.15) -- (1.15,1.15) -- (1.15,-2.15) -- (0.85,-2.15) -- cycle;
						\draw[thick,color=Magenta!70!black,rounded corners] (1.85,1.15) -- (2.15,1.15) -- (2.15,-0.15) -- (0.15,-0.85) -- (0.15,-2.15) -- (-0.15,-2.15) -- (-0.15,-0.85) -- (1.85,-0.15) -- cycle;
					\end{tikzpicture}
					\caption{Specker's triangle as a contextuality scenario.}\label{fig:H32}
				\end{subfigure}
			\end{center}
			\caption{\textbf{Specker's triangle scenario.} (a) Hypergraph representation of its formulation as a marginal scenario. Each vertex denotes a measurement, and each hyperedge a {maximal} context. There are three dichotomic measurements $X = \{A_1, A_2, A_3\}$ which are pairwise compatible, hence $\mathcal{M} = \{\{A_1, A_2\}, \{A_1, A_3\}, \{A_2, A_3\}\}$.  The outcome set is given by $O$. \newline (b) The events-based representation of the same scenario. The vertices are given by the events $(a_ia_j|ij)$ where $a_i, a_j \in O$ denote the outcomes of the measurements $A_i,A_j \in X$ from a context $\{A_i,A_j\} \in \mathcal{M}$. Blue hyperedges correspond to measuring maximal contexts, while the remaining hyperedges arise via other, less trivial, measurement protocols. The Specker's triangle corresponds to a $(3,2)$ symmetric marginal scenario, using the terminology introduced in section~\ref{se:comp}.}
			\label{fig:specktriang}
		\end{figure}
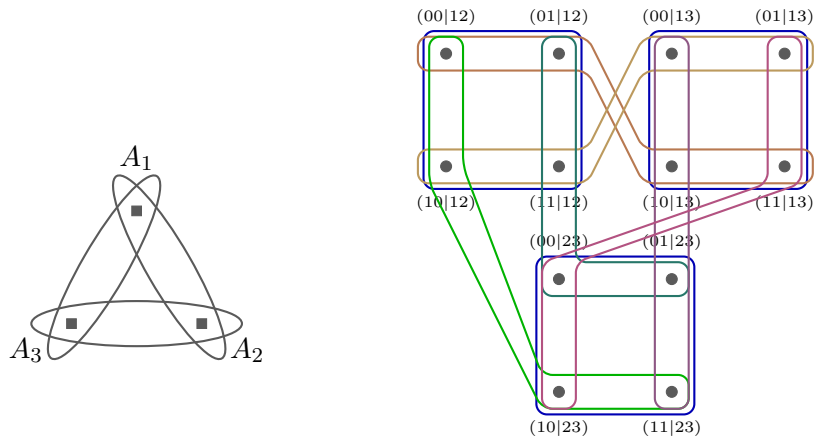

		As an example, consider the case of Specker's triangle scenario depicted in figure~\ref{fig:specktriang}. This scenario consists of three dichotomic measurements that are pairwise compatible but not triplewise compatible\footnote{For a discussion on how to realise such a scenario within quantum theory, we refer the reader to \cite{LSW}.}. Hence, as a marginal scenario $(X,O,\M)$ Specker's triangle features measurements $X=\{A_1,A_2,A_3\}$ with outcomes $O = \{0,1\}$ and maximal contexts \mbox{$\M =\{\{A_1, A_2\}, \{A_1, A_3\}, \{A_2, A_3\}\}$}, as shown in figure \ref{fig:32}.

		As an events-based hypergraph, Specker's triangle is shown in figure \ref{fig:H32}. The vertices of this hypergraph are given by the events $V = \{(a_ia_j|ij)\}_{i<j}$, where the indices $i$ and $j$ run from 1 to 3 and refer to the three measurements in $X$. One can notice that there are two types of hyperedges. First of all, there are the blue hyperedges of the form $\{(a_ia_j|ij) \}_{a_i, a_j \in O}$. These contain all the possible outcomes for the two measurements in a chosen context. Secondly, there are hyperedges corresponding to adaptive measurement protocols. As an example, consider the hyperedge depicted in dark green at the bottom left side of the hypergraph, containing vertices $(00|12)$, $(10|12)$, $(10|23)$ and $(11|23)$. This hyperedge may be understood as the following protocol. Initially, one performs measurement of $A_2$. If the outcome is $a_2 = 0$, then $A_1$ is measured next. However, if the outcome is $a_2 = 1$, then the second measurement one performs is $A_3$. The former possibility yields the events $(00|12)$ and $(10|12)$, while the latter gives $(10|23)$ and $(11|23)$.
	
	\subsection{Sets of Probabilistic Models for Events-Based Hypergraphs}\label{se:probmod}
		
		The objects of study in the events-based approach are the so-called \emph{probabilistic models}, which is a notion analogous to empirical models in the compatibility approach. A probabilistic model on a contextuality scenario $H$ is a functional $p : V \rightarrow [0, 1]$ denoting the conditional probability that an event $v$ occurs when any measurement $e$ with $v$ as one of its outcomes is performed. The identification of outcomes by operational equivalences in the definition of the contextuality scenario implies that the probability $p(v)$ is independent of the measurement $e$. Moreover, since the measurements are complete, the probabilities are normalized within each hyperedge. That is, $\sum_{v \in e} p(v) = 1$ for all $e\in E(H)$. Since the specification of a probabilistic model is a tuple of real numbers, we often treat a probabilistic model as a vector $\vec{p}$ embedded in a vector space with coordinates corresponding to the probabilities of events in $V$.
		
		We will denote the set of all valid probabilistic models on $H$ by $\G(H)$. There are various subsets of $\G(H)$ which are of particular interest. Here, we define the subsets corresponding to classical, quantum, and almost quantum models.
		
		\begin{defn}[Classical models]{\citep[Def.~4.1.1]{AFLS}}\label{def:nchv}\\
			A probabilistic model $p \in \mathcal{G}(H)$ is \emph{classical} if it can be written as a convex combination of the deterministic ones. Each deterministic probabilistic model $p \in \mathcal{G}(H)$ satisfies $p(v) \in \{0,1\}$ for all $v \in V$. The set of all classical models on $H$ is denoted by $\C(H)$ and forms a polytope, whose vertices are the deterministic probabilistic models.
		\end{defn}
		
		\begin{defn}[Quantum models]{\citep[Def.~5.1.1]{AFLS}}\label{def:qm}\\
			A probabilistic model $p \in \mathcal{G}(H)$ is \emph{quantum} if there exists a Hilbert space $\cH$, a single quantum state $\rho$ (density operator on $\cH$) and a projection operator $P_v$ on $\cH$ associated to every $v \in V$, such that 
			\begin{samepage} 
				\begin{compactenum}
					\item $\sum_{v \in e} P_v = \Id_{\cH} \quad \forall \, e \in E$,
					\item $p(v) = \tr{P_v \, \rho}\quad \forall \, v \in V$.
				\end{compactenum} 
			\end{samepage}
			The set of all quantum models on $H$ is denoted by $\Q(H)$.
		\end{defn}
		
		In the specific case of Bell scenarios the set $\Q(H)$ coincides with the traditional set of quantum correlations~\citep[Prop.~5.2.1]{AFLS}.
		\begin{defn}[Almost quantum models]{\citep[Def.~6.1.2]{AFLS}}\label{def:aqm}\\
			A probabilistic model $p \in \mathcal{G}(H)$ is \emph{almost quantum} if there exists a Hilbert space $\cH$, a single quantum state $\rho$ (density operator on $\cH$) and a projection operator $P_v$ on $\cH$ associated to every $v \in V$, such that 
			\begin{samepage} 
				\begin{compactenum}
					\item $\sum_{v \in e} P_v \leq \Id_{\cH} \quad \forall \, e \in E$,
					\item $p(v) = \tr{P_v \, \rho}\quad \forall \, v \in V$.
				\end{compactenum} 
			\end{samepage}
			The set of all almost quantum models on $H$ is denoted by $\Q_1(H)$.
		\end{defn}
		
		The terminology "almost quantum" for $Q_1$ models comes from their close connection to the set of correlations in Bell scenarios known as "almost quantum correlations"~\cite{AQ}. Indeed, it has been proven that the set $Q_1(H)$ is equivalent to the set of almost quantum correlations whenever $H$ is the hypergraph of a Bell scenario\footnote{For a discussion on how to define such Bell hypergraphs we refer the reader to \cite{AFLS,SW17}.}.
		
		A common feature shared by quantum and almost quantum models is that they both satisfy the consistent exclusivity principle~\cite{AQ, AFLS}. Explicitly, this means that all the CE inequalities defined below hold for every almost quantum (and quantum) probabilistic model.
		
		\begin{defn}[CE inequality]{\hspace{1pt}}\label{def:CE} \\
			Given a contextuality scenario $H=(V,E)$, an inequality of the form $\sum_{v \in \cS} p(v) \leq 1$ is a \emph{CE inequality} if and only if $\cS \subseteq V$ is a set of \textit{pairwise exclusive events}. {A set of events $\cS$ is pairwise exclusive if,} for each pair $u,v\in \cS$, there is a hyperedge $e \in E$ such that $u,v \in e$.
		\end{defn}
		
		{It is worth mentioning that for scenarios of the type $H[X,O,\M]$ which arise from marginal scenarios, the notion of exclusivity introduced in the definition of CE inequalities can be equivalently stated as follows~\cite{AFLS}: two events, $(\mathbf{a}|\mathbf{x})$ and $(\mathbf{a'}|\mathbf{x'})$, are exclusive if there is a measurement in $\mathbf{x} \cap \mathbf{x'}$ that assigns different outcomes in  $\mathbf{a}$ and $\mathbf{a'}$. Moreover, for these scenarios the set of empirical models on $(X,O,\M)$ is equivalent to the set of probabilistic models on $H[X,O,\M]$.}
		
		\begin{defn}[Consistent Exclusivity principle]{\hspace{1pt}} \\
			A theory satisfies the \emph{principle of consistent exclusivity} (CE) if every probabilistic model $p$ admissible by the theory satisfies all the relevant CE inequalities\footnote{A stronger version of the Consistent Exclusivity principle imposes constraints on the probabilistic models of a contextuality scenario by considering its application in larger scenarios that contain the original one (see \citep{AFLS}).}.
		\end{defn}
		
		\begin{defn}[CE models]{\hspace{1pt}}\label{def:CEmodels} \newline  
			A probabilistic model $p \in \G(H)$ is \emph{consistently exclusive} if it satisfies {the CE principle.} The set of all consistently exclusive models on $H$ is denoted by $\CE(H)$%
			\footnote{The set $\CE(H)$ defined here corresponds to the set $\CE^1(H)$ defined in \citep{AFLS}.} 
			and forms a convex polytope. 
		\end{defn}

\section{Specker's Principle and its Consequences}
	\label{sec:3}
	
	One can derive several constraints on the sets of symmetric marginal scenarios and the statistics they generate by considering consequences of Specker's principle, which we now introduce. 
	
	\begin{block}[Specker's principle as originally stated]
		"If you have several questions and you can answer any two of them, then you can also answer all of them." {\textup{[\href[pdfnewwindow]{https://vimeo.com/52923835}{vimeo.com/52923835 (2009)}, also attributed in reference~\citealp{speckerprinciple}.]}}
	\end{block}

	As stated, the principle refers to "questions", i.e.~measurements, and imposes that a set of pairwise compatible measurements is itself compatible. We can formalize this interpretation of Specker's principle as follows. 
	
	\begin{block}[Pairwise sufficiency for measurements] 
		If in a set of measurements every pair is compatible, then all the measurements are compatible.
	\end{block} 
	
	Under the principle of pairwise sufficiency for measurements, there would be no distinction between an $(n,2)$ symmetric marginal scenario and an $(n,n)$ symmetric marginal scenario, as those would merely be different representations of the same compatibility relations.
	
	This formalization of Specker's principle holds true in quantum theory for sharp \bel{(i.e.~projective)} measurements. Note that Specker's principle has elsewhere been invoked to motivate statistical constraints at the level of events, such as the principle of consistent exclusivity \cite{CSW, AFLS}. However, it is important to recognize that these statistical constraints are not equivalent to Specker's principle. Rather, they are implications thereof. In other words, Specker's principle pertains foremost to "questions" (i.e.~measurements), and only secondarily to "answers" (i.e.~events).
	
	While consistent exclusivity is a statistical constraint implied by Specker's principle, it is not the only one. We distinguish consistent exclusivity (which applies to any contextuality scenario) from the following statistical constraint (which refers to $(n,2)$ symmetric marginal scenarios only).
	
	\begin{block}[Pairwise sufficiency for probabilistic models] 
		If in a set of measurements every pair is compatible, then -- for every preparation -- the statistics generated by these measurements are marginals of some joint probability distribution.
	\end{block}
	
	Equivalently,  we can say that in a theory satisfying pairwise sufficiency for probabilistic models, the set of probabilistic models that arises from an events-based hypergraph that can be obtained from an $(n,2)$ symmetric marginal scenario for some integer $n$ is always classical. 
	
	Note that every compatible set of measurements admits a joint probability distribution over the outcomes of the measurements in the set, as implied by the definition of compatibility. Therefore, if a theory satisfies pairwise sufficiency for measurements, then it also satisfies pairwise sufficiency for probabilistic models. However, the converse need not hold, as has been pointed out at the end of section~\ref{sec:2-Compatibility}. This one-way relationship between the principle for measurements and the principle for probabilistic models is illustrated in figure \ref{fig:flowchart}.
	
	Besides the concept of pairwise sufficiency, one can define another property called all-but-one sufficiency, which might at first appear to be less constraining. All-but-one sufficiency captures the same idea as pairwise sufficiency, but applies only to a subset of marginal scenarios, namely \narrowmath{(n,n-1)} symmetric marginal scenarios. 
	
	\begin{block}[All-but-one sufficiency for measurements] 
		If in a set of at least three measurements, the elements of every proper subset are compatible, then all the measurements are compatible.
	\end{block}
	Under the principle of all-but-one sufficiency for measurements, there would be no distinction between an \narrowmath{(n,n-1)} symmetric marginal scenario and an $(n,n)$ symmetric marginal scenario for $n \geq 3$, as those would merely be different representations of the same compatibility relations. The qualifier $n\geq 3$ is needed, because otherwise all pairs of measurements would be compatible, and hence all contextuality scenarios would be deemed classical.
	
	\begin{block}[All-but-one sufficiency for probabilistic models] 
		If in a set of at least three measurements, the elements of every proper subset are compatible, then -- for every preparation -- the statistics generated by these measurements are marginals of some joint probability distribution.
	\end{block}
	In other words, we can say that a theory satisfies the principle of all-but-one sufficiency for probabilistic models if every probabilistic model $p$ is classical whenever the compatibility scenario is an \narrowmath{(n,n-1)} symmetric marginal scenario for some integer $n \geq 3$.
	
	Notice that if a theory satisfies pairwise sufficiency for measurements, then it also satisfies all-but-one sufficiency for measurements. Similarly, if a theory satisfies pairwise sufficiency for probabilistic models, then it satisfies all-but-one sufficiency for probabilistic models as well. Moreover, we now show that pairwise sufficiency and all-but-one sufficiency are actually \emph{equivalent} at the level of measurements, even though they seem to be distinct at the level of probabilistic models, as the conjecture in section \ref{sec:5} would imply, if true. The implication relations among the various consequences of Specker's principle are summarized in figure~\ref{fig:flowchart}.
	
	\begin{thm}\label{thm:notmostnot2}
		A theory satisfies the all-but-one sufficiency principle for measurements if and only if it satisfies the pairwise sufficiency principle for measurements.
	\end{thm}
	\begin{proof}
		The "if" statement holds trivially, so let us focus on the "only if" implication. Assume that a theory satisfies the all-but-one sufficiency principle for measurements. Let $X$ be a set of $n$ measurements of which every pair is compatible. We can choose $n\geq 4$, because for $n=3$, the argument is trivial. By applying the assumption, we establish that every subset of $X$ of size $3$ is compatible. This follows because every such subset is a $(3,2)$ symmetric marginal scenario. With this established, we again apply the assumption in order to establish that every subset of $X$ of size $4$ is compatible. If $n> 4$, we can further iterate applications of the assumption until we establish that all subsets of size $n$, namely the entire set $X$, is compatible.
	\end{proof}

	\begin{figure}
		\begin{center}
			\begin{tikzpicture}
				\node at (-5.65,-0.2) {\color{red!70!blue} Specker's principle};
				\node at (-3,-0.25) {\Huge{$\Longleftrightarrow$}};
				\node at (0,0) {\color{red!70!blue}Pairwise sufficiency};
				\node at (0,-0.4) {\color{red!70!blue}for measurements};
				\node at (3,-0.25) {\Huge{$\Longleftrightarrow$}};
				\node at (6.2,0) {\color{red!70!blue}All-but-one sufficiency};
				\node at (6.2,-0.4) {\color{red!70!blue}for measurements};
				\node[rotate=-90] at (-5.65,-1.45) {\Huge{$\Longrightarrow$}};
				\node[rotate=-90] at (0,-1.45) {\Huge{$\Longrightarrow$}};
				\node[rotate=-90] at (6.2,-1.45) {\Huge{$\Longrightarrow$}};
				\node at (-5.65,-2.7) {\color{PineGreen}Consistent exclusivity};
				\node at (0, -2.5) {\color{red!70!blue}Pairwise sufficiency};
				\node at (0,-3) {\color{red!70!blue}for probabilistic models};
				\node at (3.1,-2.75) {\Huge{$\Longrightarrow$}};
				\node at (6.2, -2.5) {All-but-one sufficiency};
				\node at (6.2,-3) {for probabilistic models};
			\end{tikzpicture}
		\end{center}
		\caption{Some of the consequences of Specker's principle and the implications known to hold among them. The top row contains principles pertaining to the structure of measurements, while the bottom row contains principles pertaining to sets of probabilistic models. Text colour depicts whether a given statement holds in an almost quantum theory: the green statement is satisfied, the red statements are violated, and the black statement is the subject of conjecture in section~\ref{sec:5}.}
		\label{fig:flowchart}
	\end{figure}
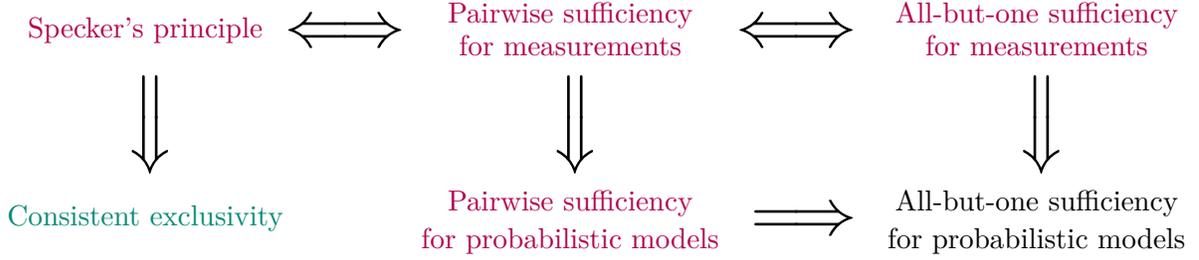

	One can naturally consider two hierarchies of principles akin to pairwise sufficiency and all-but-one sufficiency. For example, the principle of triplewise sufficiency for measurements states that triplewise compatibility implies joint compatibility. As another example, the all-but-two sufficiency principle for probabilistic models states that no nonclassical correlations can exist in a compatibility scenario wherein every collection of all-but-two measurements is compatible, i.e.~an \narrowmath{(n,n-2)} symmetric marginal scenario.

\section{Almost Quantum Models Violate Specker's Principle}
	\label{sec:4}
	
	Since projective measurements in quantum theory satisfy the principle of pairwise sufficiency, the set $\Q(H)$ for an events-based hypergraph $H$ obtained from an \narrowmath{(n,2)} symmetric marginal scenario coincides with $\C(H)$. Here, we demonstrate that there are probabilistic models in $\Q_1(H)$ for $H$ obtained from an \narrowmath{(4,2)} symmetric marginal scenario (see figure \ref{fig:42}), which are outside of $\C(H)$. This fact shows that almost quantum correlations violate "pairwise sufficiency for probabilistic models", and therefore any physical theory consistent with them violates "pairwise sufficiency for measurements" {and, by doing so, also Specker's principle.}
	
	In a \narrowmath{(4,2)} symmetric marginal scenario, each facet of the classical polytope is either a CE constraint or a pentagonal inequality~\cite{belMT}. One instance of the pentagonal inequalities reads
	\begin{align}\label{eq:pent}
		I_{\mathrm{pent}} =& - \ev{X_1X_2} - \ev{X_1X_3} + \ev{X_1X_4} + \ev{X_1} - \ev{X_2X_3} \\ \nonumber
		&+\ev{X_2X_4} + \ev{X_2} + \ev{X_3X_4} + \ev{X_3} - \ev{X_4} \leq 2\,,
	\end{align}
	where $\ev{X_iX_j} = \sum_{a_i,a_j \in O} (-1)^{a_i+a_j} p(a_ia_j|ij)$ and $\ev{X_j} = \sum_{a_j \in O} (-1)^{a_j} p(a_j|j)$. Notice that the upper bound of $2$ is satisfied for both classical and quantum probabilistic models~\cite{belMT}. This is expected, because it is well known that both classical and quantum theories respect Specker's principle.
	
	\begin{figure}[!tbh]
		\begin{center}
			\begin{tikzpicture}[scale=1.5]
				\foreach \a in {1} \foreach \b in {1,2}
				{
				\node[draw,fill, color=gray!60!black, shape=rectangle,scale=.5] at (1+\b,1-\a) {};
				\node[above=15pt] at (1+\b,1-\a) {$A_{\b}$} ;
				}
				\foreach \a in {2} \foreach \b in {1,2}
				{
				\node[draw,fill, color=gray!60!black, shape=rectangle,scale=.5] at (1+\b,1-\a) {};
				\pgfmathtruncatemacro\result{\a+\b}
				\node[below=15pt] at (1+\b,1-\a) {$A_{\result}$} ;
				}
				\draw[thick,color=gray!60!black] (2.5,0) ellipse (0.9cm and .2cm) ;
				\draw[thick,color=gray!60!black] (2.5,-1) ellipse (0.9cm and .2cm) ;
				\draw[thick,color=gray!60!black,rotate around={90:(2,-0.5)}] (2,-0.5) ellipse (0.9cm and .2cm) ;
				\draw[thick,color=gray!60!black,rotate around={90:(3,-0.5)}] (3,-0.5) ellipse (0.9cm and .2cm) ;
				\draw[thick,color=gray!60!black,rotate around={45:(2.5,-0.5)}] (2.5,-0.5) ellipse (1.2cm and .2cm) ;
				\draw[thick,color=gray!60!black,rotate around={-45:(2.5,-0.5)}] (2.5,-0.5) ellipse (1.2cm and .2cm) ;
			\end{tikzpicture}
		\end{center}
		\caption{Marginal scenario consisting of four pairwise compatible measurements. The four measurements are \mbox{$X = \{A_1, A_2, A_3, A_4\}$}, and the compatible subsets (contexts) are \mbox{$\mathcal{M} = \{\{A_1, A_2\}, \{A_1, A_3\}, \{A_1, A_4\},\{A_2, A_3\},\{A_2, A_4\},\{A_3, A_4\} \}$}.}
		\label{fig:42}
	\end{figure}
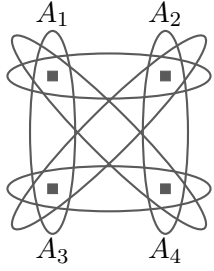
	
	A linear program optimisation over the maximum values of equation~\eqref{eq:pent} for general probabilistic models yields a value of $I_{\mathrm{pent}}^{*} = 6$. Similarly, optimising equation~\eqref{eq:pent} via a semidefinite program~\cite{AFLS} yields a maximum value of $I_{\mathrm{pent}}^{\mathrm{AQ}} = 2.5$ for almost quantum probabilistic models, which violates the classical and quantum bound of the inequality. This is despite the fact that almost quantum models satisfy all the CE inequalities.
	
	This result reveals a counterintuitive aspect of almost quantum models. Namely, there exist scenarios in which all the measurements are pairwise compatible yet almost quantum models are strictly more general than classical models. Therefore, almost quantum models are inconsistent with Specker's principle, in that any set of measurements that gives rise to them cannot satisfy it. 
	
	\begin{thm}
		\label{thm:AQCviolates2andAllBut2}
		There is no set of measurements, in any generalized probabilistic theory, which both gives rise to almost quantum models and also satisfies Specker's principle.
	\end{thm}
	
	By theorems~\ref{thm:notmostnot2} and~\ref{thm:AQCviolates2andAllBut2}, we learn that there is no set of measurements, in any theory, which both gives rise to almost quantum models and also satisfies the principle of all-but-one sufficiency for measurements.
	
{As a final comment, the research scope that tackles the formulations of general physical theories beyond the quantum one has particularly focused on the definition of a subset of measurements, known as \textit{sharp measurements}, which correspond to projective measurements in the case of quantum theory. There is still no consensus on how sharp measurements should be defined in a general theory \cite{giuliosharp,giuliosharp2}.  In the following we argue that (i) any definition of sharp measurements in a theory that yields almost quantum correlations must violate Specker's principle, and (ii) any notion of sharpness in an almost quantum theory must deviate from the candidates proposed so far \cite{giuliosharp,giuliosharp2}. }	
	
	\begin{cor}
		\label{cor:SharpNotionAntiSpecker}
		If in any theory there is a notion of sharpness, relative to which the sharp measurements yield almost quantum correlations, then there are some sharp measurements in the theory which violate Specker's principle.
	\end{cor}
	
	This corollary is an instance of theorem~\ref{thm:AQCviolates2andAllBut2}, whereby the set of measurements giving rise to almost quantum correlations is exactly the one that satisfies some (unspecified) notion of `sharpness'. 
	
	This specific case can be motivated by analogy with quantum theory. As in quantum theory, in a hypothetical almost quantum theory one must also restrict the set of allowed measurements to a strict subset of the set of all measurements in order to pick out exactly the almost quantum correlations in Kochen-Specker contextuality scenarios. If all measurements were allowed -- in either quantum or almost quantum theory -- then one could realize \emph{any} logically possible probabilistic model, using the completely noisy effects\footnote{In either case, one can always construct a measurement generating the desired set of probabilities by multiplying the unit effect of the theory with the probabilities one aims to generate.}~\cite{robustcsw}. Furthermore, the set of measurements allowed in quantum theory is exactly the set of sharp quantum measurements. In an almost quantum theory, it is unclear what this allowed set of measurements would be, but one might expect such measurements to correspond to some notion of sharpness as well. 
	
	For example, Corollary~\ref{cor:SharpNotionAntiSpecker} sheds light on the notion of sharpness for general probabilistic theories that was proposed by Chiribella and Yuan~\cite{giuliosharp,giuliosharp2}. The notion of sharpness there \emph{does} imply Specker's principle, and as such no theory can give rise to the almost quantum models using only the sharp measurements in the sense of references~\cite{giuliosharp,giuliosharp2}. In particular, such a theory cannot violate the pentagonal inequality presented in equation~\eqref{eq:pent}.

\section{Almost Quantum Models and the All-But-One Sufficiency Principle}
	\label{sec:5}
	
	The example presented in the previous section teaches us that any almost quantum theory must violate the all-but-one sufficiency principle for measurements. Incidentally, it also demonstrates that almost quantum models fail to satisfy the all-but-two sufficiency principle for probabilistic models, since $\C(H)\neq\Q_1(H)$ for an event-based hypergraph $H$ obtained from a $(4,2)$ symmetric marginal scenario. 
	
	Now, is it possible that almost quantum \emph{models} nevertheless satisfy the principle of all-but-one sufficiency, despite its violation at the level of measurements by any almost quantum \emph{theory}? It is conceivable that a physical theory might satisfy the principle of all-but-one sufficiency at the level of probabilistic models while violating the principle at the level of measurements.
	
	It might feel unnatural to divorce a statistical constraint from any restriction on the structure of measurements. Nevertheless, this unexpected satisfaction of the all-but-one sufficiency principle for probabilistic models without motivation from the corresponding measurement-based principle appears to be a feature of any almost quantum theory. 
	
	We state this as a conjecture, and provide suggestive evidence in terms of a theorem that confirms the conjecture for the special case of dichotomic measurements.
	
	\begin{conj}
		\label{conj:alloutcomes}
		Almost quantum correlations satisfy the all-but-one sufficiency principle for probabilistic models. In other words, the sets $\Q_1(H)$ and $\C(H)$ coincide, whenever $H$ corresponds to an \mbox{\narrowmath{(n,n-1)}} symmetric marginal scenario.
	\end{conj}
	
	\begin{thm}
		\label{thm:(n,n-1)}
		If $H$ is an events-based hypergraph constructed from an \narrowmath{(n,n-1)} symmetric marginal scenario with $n \geq 2$ and dichotomic measurements, then the sets $\C(H)$, $\Q(H)$, $\Q_1(H)$ and $\CE(H)$ all coincide.
	\end{thm}
	
	The proof of theorem~\ref{thm:(n,n-1)} is presented in appendix~\ref{sec:Proof}. This theorem tells us that the sets $\Q_1(H)$ and $\C(H)$ coincide whenever $H$ corresponds to an \narrowmath{(n,n-1)} symmetric marginal scenario with binary outcome measurements, in agreement with the conjecture above.
	
	Thus, we conclude that almost quantum models satisfy some statistical consequences of Specker's principle, but not others. For example, they satisfy the principle of consistent exclusivity, as well as the principle of all-but-one sufficiency for probabilistic models (at least for binary outcome measurements). On the other hand, almost quantum models satisfy neither the principle of pairwise sufficiency nor that of \emph{all-but-two} sufficiency, {be it for measurements or probabilistic models.} This dichotomy challenges our ability to understand the statistical predictions of any almost quantum theory in terms of Specker-like restrictions at the level of measurements. 

\section{Conclusions}
	\label{sec:Conclusion}
	
	In this work, we have explored some consequences of Specker's principle and their consistency with almost quantum correlations. By studying contextuality scenarios on a single system, we have found a fundamental difference between quantum theory and any potential almost quantum theory, complementing the results by Sainz et al.~\cite{noR}. There, a different no-go theorem (pertaining to almost quantum theories and the no-restriction hypothesis) was inferred based on considerations of Bell scenarios involving multiple subsystems.
	
	Our results imply that in any general probabilistic theory, the structure of measurements reproducing almost quantum models is in contradiction with Specker's principle. Accordingly, the notion of sharpness proposed in references~\cite{giuliosharp,giuliosharp2} cannot be used to recover the almost quantum correlations. This result runs counter to sentiments implicit in earlier literature. Previously, almost quantum models appeared to be the purest embodiment of Specker's principle~\cite{speckerprinciple,AQ,CSW,Yan2013,Cabello2014Tsirelson,Cabello2014FullQuantum}, in the sense that they are 
	uniquely identified\footnote{Almost quantum correlations are uniquely identified under the assumptions that the set of correlations allowed in Nature contains the quantum one, and by considering the ramifications of consistent exclusivity even in the limit of many independent copies of the scenario.} by consistent exclusivity and closure under wirings~\citep[Thrm. 7.6.2]{AFLS}. 
	
	However, consistent exclusivity is not the only constraint on probabilistic models implied by Specker's principle. Another such principle is pairwise sufficiency for probabilistic models, as defined in section~\ref{sec:3}. Despite the "success" of the  almost quantum correlations relative to consistent exclusivity, we nevertheless witness their failure to satisfy Specker's principle when analyzed through the lens of pairwise sufficiency. As such, the above preconceptions must be reversed. Almost quantum models do not exemplify Specker's principle. Rather, they are antithetical to it.
	
	For advocates of Specker's principle, our results challenge the possible physical significance of the almost quantum correlations~\cite{joe-hist}. More importantly, however, our findings restore the prominence of Specker's principle as a potential means for identifying the essence of quantum theory. The violation of Specker's principle by almost quantum models demonstrates that holistic considerations of Specker's principle enable greater insight than consistent exclusivity alone.
	
	All the results in this manuscript were found solely by studying a rather limited set of contextuality scenarios, namely symmetric scenarios with binary outcome measurements. Consideration of asymmetric compatibility structures or scenarios beyond binary-outcome measurements might illuminate stronger constraints and richer intuitions about almost quantum theory than the results presented here.
	

\section{Acknowledgements}
	We thank Emily Kendall for fruitful discussions. This research was supported by Perimeter Institute for Theoretical Physics. Research at Perimeter Institute is supported by the Government of Canada through the Department of Innovation, Science and Economic Development Canada and by the Province of Ontario through the Ministry of Research, Innovation and Science.

\setlength{\bibsep}{2pt plus 1pt minus 2pt}
\bibliographystyle{apsrev4-1}
\nocite{apsrev41Control}

\bibliography{winterschool}

\appendix

\section{Proof of Theorem \ref{thm:(n,n-1)}}
	\label{sec:Proof}
	
	For the proof of theorem~\ref{thm:(n,n-1)}, it is useful to introduce the concept of a generalized event in relation to the events as vertices of some hypergraph $H$. This notion refers to a coarse-graining of events in $H$. For instance, in the case of figure~\ref{fig:H32}, a generalized event \mbox{$g=(a_1|1)$} can be defined as either the coarse graining of $\{(a_1, a_2|1,2), (a_1, \bar{a}_2|1,2)\}$ or that of $\{(a_1, a_3|1,3), (a_1, \bar{a}_3|1,3)\}$, since the measurement $A_1$ is compatible with both $A_2$ and $A_3$ and the no-disturbance condition holds. Each of these sets $\{(a_1, a_j|ij), (a_1, \bar{a}_j|ij)\}$ is then a \textit{refinement} of~$g$. The probability assigned to a generalized event is then given by summing the probabilities of all the events in any of its refinements.
	
	The notion of exclusivity of events can be also extended to generalized events. Two generalized events $g,g'$ are said to be exclusive if events in a refinement of $g$ are pairwise exclusive with events in a refinement of $g'$. Equivalently, one can say that the union of any refinements of $g$ and $g'$ is a set of pairwise exclusive events in $V$.
	
	In this appendix we study almost quantum probabilistic models in a particular family of contextuality scenarios, which we call \textit{binary $n$-Specker} scenarios. These are the events-based hypergraphs corresponding to \narrowmath{(n,n-1)} symmetric marginal scenarios introduced in definition~\ref{def:SCS}. Moreover, we restrict our attention to cases with the measurements in $X$ each having two outcomes, so that we can choose $O = \{0,1\}$. The set of contexts $\M$ consists of every subset of $X$ with \narrowmath{n-1} elements. That is, each proper subset of $X$ is jointly measurable, but the whole set $X$ of $n$ measurements is not necessarily jointly measurable. Figure \ref{fig:32} depicts a $3$-Specker scenario and figure~\ref{fig:43} a $4$-Specker scenario.
	
	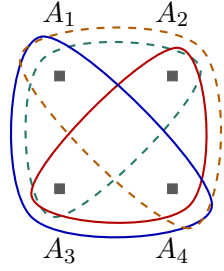
\begin{figure}[!tbh]
		\begin{center}
			\begin{tikzpicture}[scale=1.5]
				\foreach \a in {1} \foreach \b in {1,2}
				{
				\node[draw,fill, color=gray!60!black, shape=rectangle,scale=.5] at (1+\b,1-\a) {};
				\node[above=15pt] at (1+\b,1-\a) {$A_{\b}$} ;
				}
				\foreach \a in {2} \foreach \b in {1,2}
				{
				\node[draw,fill, color=gray!60!black, shape=rectangle,scale=.5] at (1+\b,1-\a) {};
				\pgfmathtruncatemacro\result{\a+\b}
				\node[below=15pt] at (1+\b,1-\a) {$A_{\result}$} ;
				}
				
				\draw[thick,color=PineGreen!70!black,dashed] plot [smooth cycle,tension=.7] coordinates {(1.85,0.15) (3.25,0.05) (1.95,-1.25) } ;
				\draw[thick,color=red!70!black, rotate around={180:(2.5,-0.5)}] plot [smooth cycle,tension=.7] coordinates {(1.85,0.15) (3.25,0.05) (1.95,-1.25) } ;
				\draw[thick,color=blue!70!black, rotate around={90:(2.5,-0.5)}] plot [smooth cycle,tension=.7] coordinates {(1.75,0.25) (3.35,0.15) (1.85,-1.35) } ;
				\draw[thick,color=orange!70!black,dashed, rotate around={-90:(2.5,-0.5)}] plot [smooth cycle,tension=.7] coordinates {(1.75,0.25) (3.35,0.15) (1.85,-1.35) } ;
			\end{tikzpicture}
		\end{center}
		\caption{The marginal scenario representation of the $\mathbf{4}$-Specker scenario, i.e.~$(4,3)$ symmetric marginal scenario. There are four dichotomic measurements $X = \{A_1, A_2, A_3, A_4\}$ which are triple-wise compatible, so that $O = \{0,1\}$, and $\mathcal{M} = \{\{A_1, A_2, A_3\}, \{A_1, A_2, A_4\}, \{A_1, A_3, A_4\}, \{A_2, A_3, A_4\}\}$.}
		\label{fig:43}
	\end{figure}
	
	From the very definition of the scenario it immediately follows that every quantum empirical model in $(X,O,\M)$ has a realisation in terms of deterministic non-contextual hidden variable models, since in quantum theory a set of pairwise compatible measurements is jointly measurable. Here we show that the set of almost quantum models has a similar behaviour. The proof relies on showing that every facet of the classical polytope corresponds to a CE inequality, each of which is satisfied by both quantum and almost quantum models.  
	
	Let us introduce some extra notation. The measurements are labelled by $X=\{1, 2, \dots, n\}$. The events in $H \coloneqq H[X,O,\M]$ are denoted by $(a_{i_1},\dots, a_{i_{n-1}})$, where $\{i_1, \dots, i_{n-1}\} \in \M$ and $a_{i_k} \in O$ is the outcome of measurement $i_k$. By $(a_{i_1},\dots, a_{i_k})$, we denote a generalized event with the remaining \narrowmath{n - k - 1} measurements in the context $\{i_1, \dots, i_{n-1}\}$ disregarded. A probabilistic model%
	\footnote{For simplicity we will denote by $p(a_{i_1}, \dots, a_{i_{k}})$ the conditional probability distribution \mbox{$p(a_{i_1}, \dots, a_{i_{k}} \mid i_1 , \dots , i_{k})$}. Similarly, the events $(a_{i_1}, \dots, a_{i_{n-1}} \mid i_1 , \dots , i_{n-1})$ in $H[X,O,\M]$ are simply denoted by $(a_{i_1}, \dots, a_{i_{n-1}})$.} 
	$\vec{p}$ is then a family of probabilities indexed by the contexts and outcomes of the measurements in the respective context
	\begin{equation}
		\label{eq:VecSp}
		\vec{p} = \bigl\{ p(a_{i_1}, \dots, a_{i_{n-1}}) \,:\, \text{$\Omega \coloneqq \{i_1, \dots, i_{n-1}\} \in \M$, $(a_{i_1}, \dots, a_{i_{n-1}}) \in O^{\Omega}$} \bigr\}.
	\end{equation}
	In this notation, the no-disturbance condition guarantees that the marginal distribution ${p(a_{i_1}, \dots, a_{i_k})}$ is independent of the choice of context from which it is computed.
	
	The rest of this appendix is dedicated to the proof of the theorem, which can written in the language of $n$-Specker scenarios as follows.
	
	\addtocounter{thmrepeat}{+2}
	\begin{thmrepeat}
		\label{thm:(n,n-1)_v2}
		Consider a binary $n$-Specker scenario $H$ with $n \geq 2$. Then the sets $\C(H)$, $\Q(H)$, $\Q_1(H)$ and $\CE(H)$ all coincide.
	\end{thmrepeat}
	
	First of all notice that there are inclusions $\C(H) \subseteq \Q(H) \subseteq \Q_1(H) \subseteq \CE(H)$, which hold in general for any contextuality scenario H~\citep{AFLS}. Therefore, in order to prove theorem~\ref{thm:(n,n-1)_v2}, we only need to show $\CE(H) \subseteq \C(H)$. This can be broken down into three parts. Firstly, we define a finite set of linear inequality constraints $\I_n$ on $\G(H)$ and the corresponding polytope $\cP$, which satisfies them. Then we show in lemma~\ref{prop:IisCE} that all these inequalities are actually CE inequalities, so that $\CE(H) \subseteq \cP$ holds. The proof concludes with lemma~\ref{prop:PisC}, which establishes the equality of $\C(H)$ and $\cP$.
	
	Consider a set of measurements $S \subseteq X$ with $k \coloneqq |S|$. Let $\mathbf{s} = \{a_r : r \in S\} \in O^S $ be a collection of outcomes of the $k$ measurements $S$. Similarly $\mathbf{o} = \{a_{r} : r \in X \setminus S\} \in O^{X \setminus S}$ is a collection of outcomes of the $n-k$ measurements $X \setminus S$. We define the set $\I_n$ in terms of linear functionals $\I_n^{k, S, \mathbf{s}, \mathbf{o}}$ on the vector space in which the probabilistic models are embedded. These depend on the choice of $S$, $\textbf{s}$ and $\textbf{o}$ in the following way.
	\begin{equation}
		\label{def:oneineqgen}
		\I_n^{k, S, \mathbf{s}, \mathbf{o}} \coloneqq p(\mathbf{o}) +  \sum_{j=1}^{k-1} (-1)^j \sum_{\{i_1,\dots, i_j\} \subseteq S} p(a_{i_1},\dots, a_{i_j}, \mathbf{o}),
	\end{equation}
	with $p(\mathbf{o}) = 1$ whenever $k=n$. We use the notation $\I_n^k$ for the set of corresponding inequalities with a given $k$.
	\begin{align}\label{def:familyineq}
		\I_n^k \coloneqq \left \{\I_n^{k,S, \mathbf{s},\mathbf{o}} \geq 0 \,:\, \text{$\mathbf{s} \in O^{S}$, $\mathbf{o} \in O^{X \setminus S}$, $S \subseteq X$, $|S| = k$} \right \}
	\end{align}
	We define the full set of inequalities, $\I_n$, as the union of all $\I_n^k$ for odd $k$ between $1$ and $n$.
	
	In fact, the right hand side of equation~\eqref{def:oneineqgen} can be expressed in simpler terms, as we illustrate next. Given a probabilistic model $\vec p$, which specifies a probability distribution for each proper subset of $X$, there is a function $f \colon O^X \to \mathbb{R}$ that recovers $\vec p$ by marginalization%
	\footnote{The marginal of $f$ over $\Omega \subseteq X$ is computed by summing the values of $f$ over $O^{\Omega}$, which yields a function $O^{X \setminus \Omega} \to \mathbb{R}$.}
	~\cite{AB}. That is, by marginalizing $f$ over any measurement $i$, we obtain the original probabilistic model for the context $X \setminus \{i\}$
	Even though this function is not unique, we can write $\I_n^{k, S, \mathbf{s}, \mathbf{o}}$ in terms of $f$ uniquely. In particular, we have
	\begin{equation}
		\label{def:oneineq}
		\I_n^{k, S, \mathbf{s}, \mathbf{o}} = f(\mathbf{s}, \mathbf{o}) +f(\bar{\mathbf{s}}, \mathbf{o}).
	\end{equation}
	Notice that one can choose $f$ to be a probability distribution over $O^X$ if and only if $\vec p$ is a classical probabilistic model.
	
	\begin{ex}
		\label{ex:3-Specker}
		As an example, consider the $3$-Specker scenario, also known as the Specker' triangle. The two sets of inequalities that give rise to $\I_3$ are
		\begin{align}
			\I_3^1 &= \left \{ p (a_i, a_j) \geq 0 \,:\, \text{$\{i,j\} \in \{1,2,3\}$, $(a_i,a_j) \in O^{\{i,j\}}$} \right \}, \\
			\I_3^3 &= \left\{ 1-p(a_1) - p(a_2) - p(a_3) + p(a_1,a_2) + p(a_1,a_3)+p(a_2,a_3) \geq 0 \,:\, a_i \in O^{\{i\}} \right\}.
		\end{align}
		The first ones are just the positivity constraints, which are always CE inequalities. Likewise, inequalities in $\I_3^3$ can be written in a form that is manifestly CE, as
		\begin{equation}
			\I_3^3 = \left\{ 1 \geq p(a_1,\bar{a}_2) + p(a_2,\bar{a}_3) + p(\bar{a}_1,a_3) \,:\, a_i \in O^{\{i\}} \right\},
		\end{equation}
		where $\bar{a}_i \coloneqq a_i + 1 \pmod 2$. A generalization of the fact that $\I_3$ consists of CE inequalities to all $\I_n$ is the content of the next lemma.
	\end{ex}
	
	\begin{lemma}
		\label{prop:IisCE}
		Every element of $\I_n$ defined as above is a CE inequality.
	\end{lemma}
	
	\begin{proof}
		We prove this result by induction on $n$. However, in order to perform the inductive step, we need the induction hypothesis to be slightly stronger than just the statement of lemma~\ref{prop:IisCE} for a particular value of $n$.
		
		\begin{inductionhypothesis}
			\hypertarget{IH}
			 For all $m$,
			\begin{compactenum}
				\item[(a)] \hypertarget{IH:a}
					every element of $\I_m^k$ is a CE inequality for any odd $k$ between 1 and $m-1$.
			\end{compactenum}
			Moreover, if $m$ is odd, then given any inequality%
			\footnote{Note that when $k$ is equal to $n$ in $\I_n^{k,S,\textbf{s},\textbf{o}}$, $S$ has to be $X$ itself and $\textbf{o}$ contains no outcomes.} 
			$\I_{m}^{m, X, \textbf{s}, \emptyset} \geq 0$ 		
			in $\I_m^m$, there exists a set $\cD_m^{\textbf{s}}$ of generalized events 
			(which by a slight abuse of notation we refer to simply as events), such that the following properties hold:
			\begin{compactenum}
				\item[(b)] \hypertarget{IH:b}
					 One can express $\I_m^{m,X,\textbf{s},\emptyset}$ as
					\begin{equation}
						\label{eq:IH2}
						\I_m^{m,X,\textbf{s},\emptyset} = 1 - \sum_{\alpha \in \cD_m^{\textbf{s}}} p(\alpha).
					\end{equation}
					
				\item[(c)] \hypertarget{IH:c} Each of the events in $\cD_m^{\textbf{s}}$ involves outcomes of at most $m-1$ measurements among $X$.
					
				\item[(d)] \hypertarget{IH:d} The set $\cD_m^{\textbf{s}}$ consists of pairwise exclusive events.
					
				\item[(e)] \hypertarget{IH:e} Each event $g$ in the set $\cD_m^{\textbf{s}}$ has the following property. There exists a measurement in $X$ whose outcome for $g$ is different from the one for $\textbf{s}$, and another measurement in $X$ whose outcome for $g$ is different from the one for $\bar{\textbf{s}}$.
			\end{compactenum}
		\end{inductionhypothesis}
	
		Notice that the statement that all elements of $\I_m$ are CE inequalities is equivalent to property~\hyperlink{IH:a}{(a)} whenever $m$ is even and equivalent to the conjunction of properties \hyperlink{IH:a}{(a)} through \hyperlink{IH:d}{(d)} whenever $m$ is odd. Therefore, Lemma~\ref{prop:IisCE} follows once we show that IH\subscript{2} and IH\subscript{3} both hold (base case), and that IH\subscript{m-2} with IH\subscript{m-1} together imply IH\subscript{m} (inductive step).
	
		\begin{basecase}
			The fact that IH\subscript{2} holds follows because $\I_2$ consists of positivity constraints only. Since 3 is odd, we need to check all five properties in order to establish IH\subscript{3}. Let's analyze them one by one.
			\begin{compactenum}
				\item[BC(a)] As we have seen in example~\ref{ex:3-Specker}, $\I_3^1$ consists of positivity constraints, which are CE inequalities.
				
				\item[BC(b)] \begin{sloppypar}Following the discussion in example~\ref{ex:3-Specker}, we can define 
				\begin{equation}
					\cD_3^{(a_1,a_2,a_3)} \coloneqq \{(a_1,\bar{a}_2), (a_2,\bar{a}_3), (\bar{a}_1,a_3)\},
				\end{equation}
				whence equation~\eqref{eq:IH2} holds for $m=3$ and $\textbf{s} = (a_1,a_2,a_3)$.\end{sloppypar}
				
				\item[BC(c)] This property is immediate from the definition of $\cD_3^{(a_1,a_2,a_3)}$ above.
				
				\item[BC(d)] For every pair of events in $\cD_3^{(a_1,a_2,a_3)}$, there is a measurement with different outcome for each of the two events in the pair. Therefore, the events in $\cD_3^{(a_1,a_2,a_3)}$ are pairwise exclusive.
								
				\item[BC(e)] Each event in $\cD_3^{(a_1,a_2,a_3)}$ is of the form $g=(a_i, \bar{a}_j)$, for some $i$ and $j$ in $\{1,2,3\}$. Hence, measurement $i$ assigns different outcomes to $g$ and $\bar{\textbf{s}}$, whereas measurement $j$ assigns different outcomes to $g$ and ${\textbf{s}}$. Property \hyperlink{IH:d}{(e)} then follows.
			\end{compactenum}
		\end{basecase}
	
		\begin{inductivestep}
			In this part of the proof of lemma~\ref{prop:IisCE}, we aim to show that the conjunction of IH\subscript{m-2} and IH\subscript{m-1} implies IH\subscript{m}. Let's check that the properties \hyperlink{IH:a}{(a)} through \hyperlink{IH:d}{(e)} are satisfied if we assume IH\subscript{m-2} and IH\subscript{m-1} in turn.
			\begin{compactenum}
				\item[IS(a)] Consider an arbitrary inequality $\I_m^{k, S, \textbf{s},\textbf{o}} \geq 0$ from $\I_m^k$, where $k$ is odd and smaller than $m$. Since $\textbf{o}$ is non-empty, we can choose a measurement $i \in X \setminus S$ and one of its outcomes $a_i \in \textbf{o}$. If we define $\textbf{o}' \coloneqq \textbf{o} \setminus \{a_i\}$, then the inequality $\I_{m-1}^{k, S, \textbf{s},\textbf{o}'} \geq 0$ is an element of $\I_{m-1}^k$. Notice that we have restricted $X$ to $X \setminus \{i\}$, but $S$ can remain unchanged, since we have chosen $i$ in a way such that $S \subseteq X \setminus \{i\}$. 
				
				It follows from IH\subscript{m-1} that $\I_{m-1}^{k, S, \textbf{s},\textbf{o}'} \geq 0$ is a CE inequality. Moreover, appending an outcome $a_i$ of one extra measurement $i$ to a set of pairwise exclusive events doesn't affect their pairwise exclusivity. Therefore, $\I_m^{k, S, \textbf{s},\textbf{o}} \geq 0$ is also a CE inequality, which completes the proof of the implication IH\subscript{m-1}$\implies$IH\subscript{m}(a).
			\end{compactenum}			
			In order to address properties \hyperlink{IH:a}{(b)} through \hyperlink{IH:d}{(e)}, we first present a candidate for $\cD_m^{\textbf{s}_m}$, given an arbitrary inequality $\I_{m}^{m, X_m, \textbf{s}_m, \emptyset} \geq 0$ in $\I_m^m$. It is defined recursively, by constructing $\cD_m^{\textbf{s}_m}$ from $\cD_{m-2}^{\textbf{s}_{m-2}}$, which corresponds to the inequality $\I_{m-2}^{m-2, X_{m-2}, \textbf{s}_{m-2}, \emptyset} \geq 0$ in $\I_{m-2}^{m-2}$. The labels for the measurements and events in the two relevant scenarios, $(X_m,O,\M_m)$ and $(X_{m-2},O,\M_{m-2})$, are
			\begin{align}
				X_m &= \{1, \dots, m\} & \textbf{s}_m &= \left( a_1,\dots, a_m \right) \\
				X_{m-2} &= \{3, \dots, m\} & \textbf{s}_{m-2} &= \left( a_3,\dots, a_m \right),
			\end{align}
			so that $\textbf{s}_{m-2}$ is just $\textbf{s}_{m}$ with the (outcomes of the) first two measurements omitted. With these defined, we can proceed to define $\cD_m^{\textbf{s}_m}$ as
			\begin{equation}
				\label{eq:cand}
				\cD_{m}^{\textbf{s}_m} \coloneqq \bigl\{(\bar{a}_1, a_2), (\bar{a}_2, \textbf{s}_{m-2}) , (a_1, \bar{\textbf{s}}_{m-2}) \bigr\} \cup \Bigl[\bigl\{(a_1, a_2), (a_1, \bar{a}_2), (\bar{a}_1, \bar{a}_2) \bigr\} \times \cD_{m-2}^{\textbf{s}_{m-2}} \Bigr].
			\end{equation}
			In the above expression, the Cartesian product has the effect of concatenating the corresponding strings of measurement outcomes. Now we can move on to proving that $\cD_m^{\textbf{s}_m}$ satisfies properties \hyperlink{IH:a}{(b)} through \hyperlink{IH:d}{(e)}.
			\begin{compactenum}
				\item[IS(b)] Recall the expression~\eqref{def:oneineq}, which corresponds to
				\begin{equation}
					\I_m^{m, X_m, \mathbf{s}_m, \emptyset} = f(\mathbf{s}_m) + f(\bar{\mathbf{s}}_m)
				\end{equation}
				in our case. If we define $\cB_m^{\textbf{s}_m} \coloneqq O^m \setminus \{ \textbf{s}_m, \bar{\textbf{s}}_m \}$, then we have%
				\footnote{This is a consequence of the function $f$ being normalized, which follows from the normalization of the probabilistic model $p$ obtained from $f$ by marginalization.}
				\begin{equation}
					\label{eq:expansion}
					\I_m^{m, X_m, \mathbf{s}_m, \emptyset} = 1 - \sum_{\mathbf{\alpha} \in \cB_m^{\textbf{s}_m}} f(\mathbf{\alpha}).
				\end{equation}     
				In order to match equation~\eqref{eq:expansion} with equation~\eqref{eq:IH2}, we identify the sum over $\cB_m^{\textbf{s}_m}$ with a sum over $\cD_m^{\textbf{s}_m}$ via the following:
				\begin{align}
					\label{eq:Step1} \scalemath{0.75}{\sum_{\mathbf{\alpha} \in \cB^{\mathbf{s}_m}_{m}}} f(\mathbf{\alpha}) &= \scalemath{0.75}{\sum_{\mathbf{\alpha} \in O^{m-2}}} \big[ f(a_1,\bar{a}_2,\mathbf{\alpha}) + f(\bar{a}_1,a_2,\mathbf{\alpha}) \big] + \scalemath{0.75}{\sum_{{\mathbf{\alpha} \in \mathcal{B}^{\mathbf{s}_{m-2}}_{m-2}}}} \big[ f(a_1,a_2,\mathbf{\alpha}) + f(\bar{a}_1,\bar{a}_2,\mathbf{\alpha}) \big] \\[-1pt] \nonumber
					& \qquad + f(\bar{a}_1,\bar{a}_2,a_3, \dots, a_m) + f(a_1,a_2,\bar{a}_3,\dots, \bar{a}_m) \\[8pt]
					\label{eq:Step2} &= \scalemath{0.75}{\sum\limits_{\mathbf{\alpha} \in \mathcal{B}^{\mathbf{s}_{m-2}}_{m-2}}} \left[ f(a_1,a_2,\mathbf{\alpha}) + f(a_1,\bar{a}_2,\mathbf{\alpha}) + f(\bar{a}_1,\bar{a}_2,\mathbf{\alpha}) \right] \\[-1pt] \nonumber
					& \qquad + f(\bar{a}_1,a_2) + f(\bar{a}_2,a_3, \dots, a_m) + f(a_1, \bar{a}_3, \dots, \bar{a}_m) \\[6pt]
					\label{eq:Step3} &= \scalemath{0.75}{\sum_{\mathbf{\beta} \in \cD_{m}^{\mathbf{s}_m}}} f(\mathbf{\beta}) \\
					\label{eq:Step4} &= \scalemath{0.75}{\sum_{\mathbf{\beta} \in \cD_{m}^{\mathbf{s}_m}}} p(\mathbf{\beta})
				\end{align}
				Equation~\eqref{eq:Step1} is just an expansion of the sum over $\cB_m^{\textbf{s}_m}$ to sums over $\cB_{m-2}^{\textbf{s}_{m-2}}$. In equation~\eqref{eq:Step2} we rearrange the terms and perform marginalization. Equation~\eqref{eq:Step3} uses the fact that the set of events being summed over is precisely $\cD_{m}^{\textbf{s}_m}$. Moreover, we use property \hyperlink{IH:b}{(b)} for $\cD_{m-2}^{\textbf{s}_{m-2}}$, which is a part of IH\subscript{m-2}. Finally, equation~\eqref{eq:Step4} replaces $f$ with $p$. This is possible, because  $\cD_{m}^{\textbf{s}_m}$ satisfies property~\hyperlink{IH:c}{(c)}, which is justified in the following step.
				
				\item[IS(c)] Since $\cD_{m-2}^{\textbf{s}_{m-2}}$ involves outcomes of no more than \narrowmath{m-3} measurements, it is easy to verify that $\cD_{m}^{\textbf{s}_{m}}$ refers to events with at most \narrowmath{m-1} measurement outcomes.
				
				\item[IS(d)] The fact that $\cD_{m-2}^{\textbf{s}_{m-2}}$ satisfies both properties \hyperlink{IH:d}{(d)} and \hyperlink{IH:e}{(e)} implies that $\cD_{m}^{\textbf{s}_{m}}$ defined by equation~\eqref{eq:cand} consists of pairwise exclusive generalized events. One can take a step further and take $\cD_{m-2}^{\textbf{s}_{m-2}}$ to be the union of the refinements of the generalized events that compose it, without changing its relevant properties. By doing this, it follows that $\cD_{m}^{\textbf{s}_{m}}$ is a set of pairwise exclusive events in $H[X_m,O,\M_m]$.
				
				\item[IS(e)] Direct inspection shows that every event in $\cD_{m}^{\textbf{s}_{m}}$ comprises a measurement that assigns a different outcome to $\textbf{s}_m$, and one that assigns a different outcome to $\bar{\textbf{s}}_m$. Hence, $\cD_{m}^{\textbf{s}_{m}}$ satisfies property \hyperlink{IH:e}{(e)} as well.			
			\end{compactenum}
		Thus the proof that IH\subscript{m-2} and IH\subscript{m-1} imply IH\subscript{m} is complete. As a consequence, the induction hypothesis holds for all integers greater than 1, which proves lemma~\ref{prop:IisCE}.
		\end{inductivestep}
	\vspace{-25pt}
	\end{proof}
	
	An immediate corollary of lemma~\ref{prop:IisCE} is the inclusion $\CE(H) \subseteq \cP$. The rest of the proof of theorem~\ref{thm:(n,n-1)_v2} consists of demonstrating $\cP = \C(H)$.

	\begin{lemma}\label{prop:PisC}
		The set of probabilistic models satisfying all inequalities in $\I_n$ coincides with the polytope of classical probabilistic models, i.e.~$\cP=\C(H)$.
	\end{lemma}
	
	\begin{proof} 
		The proof is divided into three parts, each corresponding to one of the claims below. Firstly, we establish the dimension of the polytopes, showing that it is the same for both polytopes. Next, we show that every facet of $\cP$ contains a facet of $\C(H)$. Finally, we prove that $\C(H)$ does not have any facets \emph{other} than those contained within facets of $\cP$. These facts together imply $\cP=\C(H)$.
	
		\begin{claim}
			\label{prop:C(H)dim}
			The dimension of both $\C(H)$ and $\cP$ is \narrowmath{2^n - 2}.
		\end{claim}
	
	
		Notice that every element of the probabilistic model $p \in \G(H)$ can be written as a linear combination of terms that refer to outcome(s) $0$ only. There are \narrowmath{\sum_{j=1}^{n-1} {{n} \choose{j}} = 2^n-2} such terms, which gives an upper bound on the dimensions of $\G(H)$, $\cP$ and $\C(H)$.
		
		One simple consequence of lemma~\ref{prop:IisCE} is that $\C(H)$ is a subset of $\cP$. Therefore, it suffices to show \narrowmath{\mathrm{dim}(\C(H)) = 2^n-2}.
		Note that the dimension of the classical polytope in an $n$-Specker scenario is precisely one less than the number of degrees of freedom in a joint probability distribution of $n$ binary outcome measurements, i.e.~one less than \narrowmath{2^n-1}~\cite{CollinsGisinNotation}. Hence, the classical polytope has dimension \narrowmath{2^n-2}, as required.
		
		\begin{claim}
			\label{facetinclusion}
			Every facet of $\cP$ contains a facet of $\C(H)$. Moreover, the facets of $\cP$ are in one-to-one correspondence with the inequalities $\I_n$.
		\end{claim}
		
		Generically, given a polytope $\mathcal{R}$, a linear inequality $\mathcal{A} \geq 0$ corresponds to a facet of $\mathcal{R}$ if and only if $\mathrm{dim(\mathcal{R})}$ or more vertices of $\mathcal{R}$ lie within the hyperplane $\mathcal{A} = 0$. Moreover, if the above condition is satisfied, the facet of $\mathcal{R}$ lies within the hyperplane $\mathcal{A} = 0$.
		
		Now, to every facet of $\cP$, we can associate an inequality from $\I_n$. Furthermore, each inequality in $\I_n$ is saturated by a set of exactly \narrowmath{2^n-2} vertices of $\C(H)$, and no two of these sets are the same. This fact can be seen easily by considering the form of inequalities in equation~\eqref{def:oneineq} in terms of $f$.
	
		Since \narrowmath{2^n-2} is also the dimension of $\C(H)$ by claim~\ref{prop:C(H)dim}, each inequality in $\I_n$ is associated with a different facet of $\C(H)$. Therefore, every facet of $\cP$ contains some facet of $\C(H)$, as $\C(H) \subseteq \cP$. This relation between $\C(H)$ and $\cP$ is depicted in figure~\ref{fig:oct}. Moreover, every inequality in $\I_n$ corresponds to a different facet of $\cP$. The number of facets of $\cP$ is thus equal to the number of distinct inequalities in $|\I_n|$, which is
		\begin{equation}
			\vert \I_n \vert = \frac{1}{2} \sum_{k \in \Omega} {{n}\choose{k}} \, 2^k \,2^{n-k} = 2^n \, 2^{n-2} = 2^{2n-2},
		\end{equation}
		where $\Omega$ denotes the set of odd integers between 1 and $n$ and the factor of ${\scriptstyle \nicefrac{1}{2}}$ takes into account the fact that \mbox{$\I_n^{k, S, \mathbf{s}, \mathbf{o}} = \I_n^{k, S, \bar{\mathbf{s}}, \mathbf{o}}$}, so that the corresponding inequality is counted only once.
	
	\begin{figure}[!tbh]
		\begin{center}
			\begin{tikzpicture}
				\draw[thick, color=gray!60!black] (0,0) rectangle (3,3);
				\draw[thick, color=blue!60!black] (1,0) -- (2,0) -- (3,1) -- (3,2) -- (2,3) --  (1,3) -- (0,2) -- (0,1) -- cycle;
				\node[draw=gray!60!black,fill, color=blue!60!black, shape=circle,scale=.5] at (1,0) {};
				\node[draw=gray!60!black,fill, color=blue!60!black, shape=circle,scale=.5] at (2,0) {};
				\node[draw=gray!60!black,fill, color=blue!60!black, shape=circle,scale=.5] at (3,1) {}; 
				\node[draw=gray!60!black,fill, color=blue!60!black, shape=circle,scale=.5] at (3,2) {}; 
				\node[draw=gray!60!black,fill, color=blue!60!black, shape=circle,scale=.5] at (2,3) {}; 
				\node[draw=gray!60!black,fill, color=blue!60!black, shape=circle,scale=.5] at (1,3) {}; 
				\node[draw=gray!60!black,fill, color=blue!60!black, shape=circle,scale=.5] at (0,2) {}; 
				\node[draw=gray!60!black,fill, color=blue!60!black, shape=circle,scale=.5] at (0,1) {};
				\node at (1.5,1.5) {\color{blue!60!black}{$\mathcal{R}'$}};
				\node at (-0.3,3.3) {$\mathcal{R}$};
			\end{tikzpicture}
		\end{center}
		\caption{Two polytopes, $\mathcal{R}$ (gray) and $\mathcal{R}'$ (blue). $\mathcal{R}'$ is contained within $\mathcal{R}$, and each facet of $\mathcal{R}$ contains one of $\mathcal{R}'$. This schematically depicts the relation between the classical polytope $\C(H)$ (which plays the role of $\mathcal{R}'$) and the polytope $\cP$ (which plays the role of $\mathcal{R}$) proven in Claim \ref{facetinclusion}.}
		\label{fig:oct}
	\end{figure}
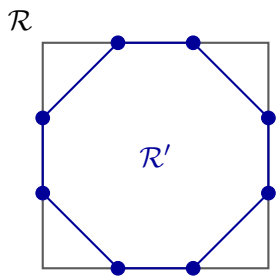
			
		\begin{claim}
			\label{nootherfacets}
			The number of facets of $\C(H)$ does not exceed the number of facets of $\cP$. Together with Claim~\ref{facetinclusion}, this implies that $\C(H)$ does not have any facets apart from those contained within facets of $\cP$. 
		\end{claim}
	
		The number of vertices of $\C(H)$ is $N \coloneqq \vert O^n \vert = 2^n$. These correspond to all the possible deterministic assignments of probabilities to a joint outcome of the $n$ measurements. As a matter of fact, no polytope with $2^n$ vertices and dimension \narrowmath{2^n-2} has more than $2^{2n-2}$ facets, which happens to be the number of facets of $\cP$.
		
		In greater generality, McMullen~\cite{mcmuellen} showed that any convex polytope with a given dimension $d$ and number of vertices $N$ has at most as many faces as the cyclic polytope $C(d,N)$. Let us denote the number of facets of $C(d,N)$ by $\#_{d-1} (C(d,N))$. Theorem 4 by Gale~\cite{gale} asserts that for even dimensions, $d = 2m$, this number satisfies
		\begin{equation}
			\#_{2m-1}\big( C(2m,N) \big) = {N-m \choose m} + {N-m-1 \choose m-1}.
		\end{equation}
		Hence, for $d = 2^n - 2$ and $N= 2^n$, we have the following.
		\begin{align}
			\#_{d-1}\big( C(d,2^n) \big) &= {2^{n-1}+1 \choose 2^{n-1}-1} + {2^{n-1} \choose 2^{n-1}-2} \\
			&= \frac{1}{2} \left [ \bigl(2^{n-1}+1\bigr) 2^{n-1} + 2^{n-1} \bigl( 2^{n-1} - 1 \bigr) \right ] \\
			&= 2^{2n-2}
		\end{align}
		Consequently, the number of facets of $\C(H)$ cannot exceed $\#_{d-1}(C(d,2^n)) = 2^{2n-2}$, as required. It follows that $\C(H)$ has no other facets than those it shares with $\cP$ and thus $\C(H)=\cP$.
	\end{proof}
	
	Lemmas \ref{prop:IisCE} and \ref{prop:PisC} provide a complete characterisation of how quantum and almost quantum behaviours are realized in binary $n$-Specker scenarios. On the one hand, we have $\C(H) \subseteq \Q(H) \subseteq \Q_1(H)\subseteq \CE(H)$ as a general hierarchy for all scenarios. On the other hand, for binary $n$-Specker scenarios, $\CE(H) \subseteq \cP$ and $\cP = \C(H)$, so that all the inclusions are actually equalities, concluding the proof of theorem~\ref{thm:(n,n-1)_v2}.
	
	Hence, binary $n$-Specker scenarios have the property that every quantum and almost quantum probabilistic model has a realisation in terms of noncontextual deterministic hidden variable models. That is, both quantum and almost quantum statistics satisfy the \mbox{\narrowmath{(n-1)}-sufficiency} principle for probabilistic models within binary $n$-Specker scenarios, proving theorem~\ref{thm:(n,n-1)}.
	
	As a final remark, notice that theorem~\ref{thm:(n,n-1)_v2} is not trivial, in the sense that binary $n$-Specker scenarios do admit general probabilistic models that violate consistent exclusivity~\cite{LSW}.

\end{document}